\documentclass[12pt,a4paper,oneside]{article}
  \usepackage{authblk} %for writing authors
  \usepackage{url}
  \usepackage{amsthm}
  \usepackage{amsmath}
  \usepackage{amsfonts}
  \usepackage{amssymb}
  \usepackage{amscd}
  \usepackage{authblk}
  \usepackage{bm}
  \usepackage[mathscr]{eucal}%for \mathscr font
  \usepackage{mathtools} % for 'bsmallmatrix' environment
  \usepackage{enumitem}
  \usepackage{physics}%for using Dirac's bra-ket notation, code: $\bra{\Psi}\ket{\Psi}$, $\expval{A}{\Psi}$, $\braket{x}{y}$
  \usepackage{color}
  \usepackage{fancyhdr}
  \usepackage{hyperref}
  \usepackage{marginnote}%for comments in margin, use \marginpar{text}
  \usepackage[top=2.5cm, bottom=2.7cm,left=2.5cm, right=2.5cm, marginparwidth=1.8cm]{geometry}
  \usepackage{framed}% use \begin{frame} text \end{frame} to put it in a framed box
  \usepackage[noabbrev]{cleveref} %for numbering in align* use \cref{...} for refering
  \usepackage[T1]{fontenc} %for writing multiple authors with multiple affiliations
 \usepackage[utf8]{inputenc}%for writing multiple authors
 \usepackage{blkarray}%for block matrices with subspaces marked for example see below

\usepackage{relsize} % example: \mathlarger{\mathlarger{\Pi}}
   \usepackage{tikz}
 \usepackage{forest}
 \usepackage{qtree}  
 
  \numberwithin{equation}{section}
 
 \allowdisplaybreaks[1] 
  
  \theoremstyle{definition}  %Title and number in bold, body in normal font.
   \newtheorem{defn}{Definition}[section]
   
   \newtheorem{egs}[defn]{Examples}

   \newtheorem{rmk}[defn]{Remark}

  \theoremstyle{plain}  %n Title and number in bold, body in italic (default).
   \newtheorem{thm}[defn]{Theorem}
   \newtheorem{lem}[defn]{Lemma}
   \newtheorem{prop}[defn]{Proposition}

  \theoremstyle{remark} %Title and number in italic, body in normal font.

   \newcommand{\B}[1]{\mathscr{B}({#1})}
   
   \newcommand{\Bt}[1]{\mathscr{B}_2({#1})}

   \newcommand{\CH}{\mathcal{H}}
   \newcommand{\CK}{\mathcal{K}}

   \newcommand{\BR}{\mathbb{R}}
   
   \newcommand{\BN}{\mathbb{N}}

\usepackage{centernot}
\newcommand{\nll}{\centernot{\ll}}

 \newcommand{\numberthis}{\refstepcounter{equation}\tag{\theequation}}  
\makeatletter
\newcommand\footnoteref[1]{\protected@xdef\@thefnmark{\ref{#1}}\@footnotemark}
\makeatother

\newcommand{\emailaddress}[1]{\newline{\sf#1}}

\let\tr\relax 
\DeclareMathOperator{\tr}{Tr}

\let\ker\relax 
\DeclareMathOperator{\ker}{Ker}
\DeclareMathOperator{\ran}{Ran}

\DeclareMathOperator{\supp}{Supp}

\DeclareMathOperator{\spn}{span}

\fontsize{12}{14}

\title{Quantum \texorpdfstring{\ensuremath{f}}{}-divergences via  Nussbaum-Szko{\l}a Distributions and Applications to \texorpdfstring{\ensuremath{f}}{}-divergence Inequalities}
\author[1]{George Androulakis}
\author[1,2]{Tiju Cherian John}

\affil[1]{University of South Carolina, Columbia, South Carolina, USA }
\affil[2]{The University of Arizona, Arizona, USA
\emailaddress{giorgis@math.sc.edu, tijucherian@fulbrightmail.org}
}

\begin{document}
\maketitle
%\tableofcontents
%\newpage
\begin{center}
    In memory of K. R. Parthasarathy
\end{center}

 \begin{abstract}
The main result in this article shows that the quantum $f$-divergence of two states  is equal to the classical $f$-divergence of the corresponding Nussbaum-Szko{\l}a distributions. This provides a general framework for studying certain properties of quantum entropic quantities using the corresponding classical entities. The usefulness of the main result is  illustrated by obtaining several quantum $f$-divergence inequalities from their classical counterparts. 
All results presented here are valid in both finite and infinite dimensions and hence can be 
 applied to continuous variable systems as well. A comprehensive review of the instances in the literature where  Nussbaum-Szkoła distributions  are used, is also provided in this article.
 \end{abstract}
 
 \section{Introduction to Quantum Entropic Quantities}\label{sec:introduction}
Following the pioneering work of  Shannon on information theory, Kullback and Leibler defined  a divergence of two probability measures $P$ and $Q$, using the formula \cite{kullback-leibler-1951}\begin{align*}
    D(P\|Q)=   \begin{cases}
        \sum_{i} P(i)\log \frac{P(i)}{Q(i)}, & \text{if } P\ll Q \text{ with the convention } 0\log \frac{0}{0}=0;\\
        \infty, &\text{otherwise,}
    \end{cases} 
    \end{align*} 
This is now known as the Kullback-Leibler divergence of probability measures. Later on,  R\'enyi \cite{renyi1961}  showed that the  Kullback-Leibler divergence can be extended to a class of relative entropic quantities $D_\alpha$ for $\alpha\in (0,1) \cup (1,\infty)$,  by defining
 \begin{align*}
       D_{\alpha}(P||Q)  = \begin{cases}
           \frac{1}{\alpha-1} \log\sum_i P(i)^{\alpha}Q(i)^{1-\alpha}, &\text{if } \alpha<1, 
            \text{ or }
            %\alpha>1 \text{ and } 
            P\ll Q ;\\
            \infty, & \text{otherwise,}
        \end{cases}
    \end{align*}
 with the convention  $0\cdot \infty=0$.
Then  $\lim_{\alpha\rightarrow 1} D_{\alpha}(P||Q) = D(P||Q).$ In the same article \cite[Page 561]{renyi1961}
  R\'enyi introduced another quantity which is now  known as $f$-divergence. This has been studied extensively in the early years in \cite{Morimoto1963, csiszar1964, Ali-Silvey-1966}. One may refer to  \cite{Liese-Vajda-2006, Csiszar-Sheilds-2004}  for modern treatments. Both Kullback-Leibler and Rényi $\alpha$-relative entropy, as well as    many other important relative entropic quantities such as Hellinger, $\chi^2$, and total variation distance  are particular cases of $f$-divergences.

  Quantizing these developments, Umegaki  extended the notion of Kullback-Leibler divergence to the setting of von Neumann algebras \cite{umegaki-1962}, and Petz  defined general quasi-entropies for quantum systems, (which includes quantum $f$-divergences) \cite{petz1986quasientropy}. 
 
 Quantum $f$-divergences have specific operational meaning especially because they satisfy the data processing inequality \cite{Hiai-2018}. Specific cases of quantum $f$-divergences like Petz-R\'enyi $\alpha$-relative entropy are useful in quantum hypothesis testing and related areas \cite{Nussbaum-Szkola-2009, Mosonyi-2009, Jaksic-Ogata-Pillet-2012,  Ses-Lam-Wil-2018, Berta-Scholz-Tomamichel-2018, Hiai-2018}. One may refer to \cite{petz-hiai-1991} and the other articles cited above to know more about different uses of relative entropy. Similar to the classical setting,  a number of quantum entropic quantities arise from quantum $f$-divergences. Thus quantum $f$-divergence is a parent quantity for several useful relative entropic quantities. Some examples are pictorially represented below.

\vspace{1cm}
\Tree[.{\textbf{Quantum} $f$-\textbf{divergence}}    [.{Umegaki RE}  ] [.{Petz-R\'enyi $\alpha$-RE}   ] [.{{Hellinger}} ] [.{$\chi^2${}} ] [.{Total Variation} ]] 
\vspace{1cm}

 The article \cite{Nussbaum-Szkola-2009} cited above is particularly important in our context. In that article, Nussbaum and Szkoła introduced classical probability distributions, (now called Nussbaum-Szkoła distributions), corresponding to any two  finite dimensional density matrices. They also proved that the Petz-Rényi $\alpha$-relative entropy of two quantum states is equal to the Rényi $\alpha$-relative entropy of the corresponding Nussbaum-Szkoła distributions. Furthermore, these distributions were extensively used in the literature to study various properties of Petz-Rényi $\alpha$-relative entropy.

In the context of our work in this article, most important related works were produced by Hiai and Mosonyi \cite{Hiai-Mosonyi-2017}, Hiai \cite{Hiai-2018}, and Berta, Scholz and Tomamichel \cite{Berta-Scholz-Tomamichel-2018}  where quantum $f$-divergences are studied in general or the more specific case of Petz-R\'enyi relative entropy is studied. Also, Seshadreeshan, Lami and Wilde in \cite{Ses-Lam-Wil-2018} and Parthasarathy in \cite{Par2021a} studied the Petz-R\'enyi relative entropy.  

The structure of this article is as follows. In Section \ref{sec:review}, we provide a comprehensive review of the literature on the use of  Nussbaum-Szkoła distributions. In Section \ref{sec:f-divergence}, we  prove that the quantum $f$-divergence of two states on a finite or infinite dimensional Hilbert space is equal to the corresponding classical $f$-divergence of the Nussbaum-Szkoła distributions (see Theorem \ref{thm:f-divergence}). This is a far reaching generalization of the result by Nussbaum and Szkoła that we described above.  In Section \ref{sec:quant-f-div-inequalities} we show that several inequalities between various quantum relative entropic quantities follow immediately from their classical counterparts, because of our Theorem \ref{thm:f-divergence}. We collect all the necessary background materials in the Appendix. 
    
\section{A Review of the Literature on Nussbaum-Szkoła Distributions}  \label{sec:review}

In this section, we present a review of the use of Nussbaum-Szkoła distributions in the literature. 
These
distributions were originally introduced in order to study the error exponents in quantum hypothesis testing problems. We do not present a general review of the quantum hypothesis testing problem which can be found in \cite{Spehner2014-oh, Bae2015-hb}. Also, we do not present a general review of $f$-divergences which can be found in the following references \cite{Hiai2011-qb, Hiai2017-sg, Hiai-Mosonyi-2017, Hiai-2018, Hiai2019-xn, hiai2021quantum}.

In binary hypothesis testing, one is presented with a null hypothesis $H_0$ and an alternative hypothesis $H_1$ with priors $\pi_0$ and $\pi_1$, respectively. There are two types of errors associated with the  distinguishability of $H_0$ and $H_1$: 
\begin{quote}
    Type I error or ``false positive'' or ``false alarm'' is when one decides $H_1$ while $H_0$ is the reality, and \\
    Type II error or ``false negative'' or ``missed detection'' is when one decides $H_0$ while $H_1$ is the correct hypothesis.
\end{quote}
 
Nussbaum and Szkoła \cite{Nussbaum-Szkola-2009} study the symmetric binary quantum hypothesis testing. In this case, Type I and Type II errors are treated equally and one is interested in minimizing the total probability of error.  Given two quantum states $\rho$ (null hypothesis) and $\sigma$ (alternate hypothesis) with priors  $\pi_0$ and $\pi_1$, respectively, one considers $n$ many i.i.d. copies $\rho^{\otimes n}$ and $\sigma^{\otimes n}$ and finds a Positive Operator Valued Measurement (POVM), $\{1-\Pi_n, \Pi_n\}$  which discriminates the states $\rho^{\otimes n}$ and $\sigma^{\otimes n}$. If  $1-\Pi_n$ is the result of the   measurement, then we decide that the given state is $\rho^{\otimes n}$, and if $\Pi_n$ is the result of the   measurement, then we decide that the given state is $\sigma^{\otimes n}$. The integer $n$ represents the number of repetitions of the experiment and we consider the sequence of measurements $\Pi = \{\Pi_n\}_{n\in \BN}$. The probability of Type I and Type II errors when using $n$ copies of the states are given by \begin{equation}\label{eq:error-types}
    \alpha_n(\Pi) := \tr \left[\Pi_n\rho^{\otimes n}\right] \quad \text{and}  \quad \beta_n(\Pi):=\tr \left[\left(1-\Pi_n\right) \sigma^{\otimes n}\right],
\end{equation} respectively.
The Bayesian error probability is given by 
\begin{align}\label{eq:error-pi-n}
    \operatorname{Err}(\Pi_n) &= \pi_0 \alpha_n(\Pi)+ \pi_1\beta_n(\Pi) \nonumber\\
    &= \pi_1 -\operatorname{Tr}\left[\Pi_n\left(\pi_1\sigma^{\otimes n}-\pi_0\rho^{\otimes n}\right)\right].
\end{align} It is known that the optimal measurement (which minimizes the error probability in \eqref{eq:error-pi-n}) is given by the Holevo-Helstrom projection \cite{Holevo-1978, helstrom-1976},
\[\Pi_n^* = \supp (\pi_1\sigma^{\otimes n}-\pi_0\rho^{\otimes n})^+,\]
where \lq\lq$\supp$\rq\rq\ denotes the support projection and  the superscript \lq\lq$+$\rq\rq\ denotes the positive part, i.e., \[P_{e, \min, n} = \operatorname{Err}(\Pi_n^*).\] 
Hence, by using the facts $\left(\supp{A^+}\right)A = A^+$ and $A^+ = \frac{A+\abs{A}}{2}$, the minimum of the Equation~\eqref{eq:error-pi-n})  is equal to 
\[ P_{e, \min, n} =  \frac{1}{2}\left(1-\norm{\pi_1\sigma^{\otimes n}-\pi_0\rho^{\otimes n}}_1\right) \]
where $\norm{\cdot}_1$ denotes the trace class norm.
The main result in \cite{Nussbaum-Szkola-2009} is the \textbf{optimality} part of the error exponent given by the inequality 
 
\begin{align}\label{eq:n-s-chernoff}
   \liminf _{n \rightarrow \infty} \frac{-1}{n} \log P_{e, \min, n} \leq -\inf _{0 \leq s \leq 1} \log \operatorname{Tr}\left[\rho^{1-s} \sigma^s\right],
\end{align}
where the minus signs ensure that the quantities we compare are positive.
In order to obtain this bound, Nussbaum and Szkoła introduce two classical probability distributions $P$ and $Q$ corresponding to the given states $\rho$ and $\sigma$ and use tools of classical estimation theory. These distributions $P$ and $Q$ are now known as the Nussbaum-Szkoła distributions. A formal definition of the Nussbaum-Szkoła distributions can be seen in our Definition~\ref{defn:nussbaum-szkola}. In \cite{Audenaert2007-gl}, Audenaert et al. prove the opposite inequality of \eqref{eq:n-s-chernoff} (\textbf{achievability}) and hence we have an equality there. The quantity on the right hand side of \eqref{eq:n-s-chernoff} is called the \textbf{quantum Chernoff bound}. Thus the optimal error exponent in the symmetric quantum hypothesis testing is obtained by combining the results in the articles \cite{Nussbaum-Szkola-2009, Audenaert2007-gl} and it is given by 
\begin{align}\label{eq:chernoff}
   C(\rho, \sigma) := \liminf _{n \rightarrow \infty} \frac{-1}{n} \log P_{e, \min, n} = -\inf _{0 \leq s \leq 1} \log \operatorname{Tr}\left[\rho^{1-s} \sigma^s\right].
\end{align}
This concludes the complete solution to the symmetric version of the binary quantum hypothesis testing problem that was originally studied by  Ogawa and Hayashi \cite{Ogawa-Hayashi-2004}. This result has been extended from states on finite dimensional Hilbert space to states on infinite spin chains satisfying certain factorization properties by Hiai, Mosonyi and Ogawa in \cite{Hiai-mosonyi-ogawa-2007}. Furthermore, Nussbaum and Szkoła \cite{Nussbaum2010-jz} consider the problem of discriminating multiple quantum states of an infinite spin chain and obtain an asymptotic bound on the error probability in increasing block size. Also,  Nussbaum and Szkoła \cite{Nussbaum2011-uo} investigate the problem of quantum hypothesis testing for multiple states on a finite dimensional Hilbert space that have pairwise disjoint supports. They conjecture that the asymptotic error exponent $\xi$  is given by 
the minimum over all pairwise Chernoff bounds between any two states in the collection, i.e, 
\[C\left(\rho_1, \ldots, \rho_r\right):=\min _{(i, j): i \neq j} C\left(\rho_i, \rho_j\right),\] known as the \textbf{multiple quantum Chernoff bound} and prove that 
$\frac{1}{3}C\left(\rho_1, \ldots, \rho_r\right)\leq \xi\leq C\left(\rho_1, \ldots, \rho_r\right)$. Audenaert and Mosonyi \cite{audenaert-mosonyi-2014} study upper bounds on the error exponents  and also prove that  $\frac{1}{2}C\left(\rho_1, \ldots, \rho_r\right)\leq \xi\leq C\left(\rho_1, \ldots, \rho_r\right)$. Finally, Li \cite{Li2016-tn} proves that \[\xi = C(\rho_1, \rho_2,\dots, \rho_r),\] which settles the conjecture of Nussbaum and Szkoła. It may be noted that
 Nussbaum-Szkoła distributions are handy in the computations leading to  the optimality of the error exponent in the problem of 
multiple quantum state discrimination which is studied in the above references.

In the asymmetric case, the Type I and Type II errors are not treated equally. For example, in the detection of COVID-19 (or most medical diagnoses), $H_0$ represents an undesired hypothesis and we want to avoid the Type II error (false negative) as much as possible. Let the minimum of the Type II errors for all possible measurements for which the corresponding Type  I error is bounded above by some $0<\epsilon<1$ be denoted by $\beta_n^*(\epsilon)$. The error rate, $\beta_R(\epsilon)$ of the Type II error is defined as \[\beta_R(\epsilon)=\lim_{n \to \infty}-\frac{1}{n}\log \beta_n^*(\epsilon).\] The \textbf{quantum Stein's Lemma} that was first proved by Hiai and Petz \cite{petz-hiai-1991} states that $\beta_R(\epsilon)$ is independent of $\epsilon$ and it is equal to $D(\rho||\sigma)$, where $D(\rho|| \sigma)$ denotes the Umegaki relative entropy of  state $\rho$ from state $\sigma$. 

Another important error bound we consider in the asymmetric quantum hypothesis testing is the \textbf{quantum Hoeffding bound}. The rates of Type I and Type II errors are defined respectively as \begin{equation}
    \alpha_R(\Pi) = \lim_{n\to\infty} -\frac{1}{n}\log \alpha_n(\Pi) \quad \text{and} \quad \beta_R(\Pi) = \lim_{n\to\infty} -\frac{1}{n} \log \beta_n(\Pi), 
\end{equation} where $\alpha_n(\Pi)$ and $\beta_n(\Pi)$ have been defined  in \eqref{eq:error-types}.
The quantum  Hoeffding bound for discriminating  two states $\rho$ and $\sigma$ relevant to a positive parameter $r$ is defined to be the 
supremum of the rate $\alpha_R(\Pi)$ for all POVM's $\Pi$ that satisfy $\beta_R(\Pi) \geq r$. 
An expression for the quantum Hoeffding bound has been proposed by  Ogawa and Hayashi \cite{Ogawa-Hayashi-2004}. It is verified by Hayashi \cite{Hayashi2007-cp} that indeed the quantum Hoeffding bound is bounded below by the proposed expression, (achievability). It is verified by Nagaoka \cite{Nagaoka2006-ci} that the quantum Hoeffding bound is bounded above by the 
same proposed expression (optimality). Thus the quantum Hoeffding bound is equal to the proposed expression of Ogawa and Hayashi.
This equality is now known as the quantum Hoeffding bound Theorem. All these ideas are combined and generalized by  Audenaert, Nussbaum, Szkoła and Verstraete in \cite{Audenaert-et-al-2008}. Furthermore, Audenaert, Mosonyi and Verstraete \cite{Audenaert2012-mo} compute the Chernoff, Hoeffding, and Stein bounds as well as the mixed error probabilities related to Chernoff and Hoeffding bounds of finite sample size. Jakšić, Ogata, Pillet and Seiringer \cite{Jaksic-Ogata-Pillet-2012} consider the binary quantum hypothesis problem on general von Neumann algebras and they find that the optimal error exponents are given by similar formulas even in this general setting.
In all of these proofs given in \cite{Nussbaum-Szkola-2009},  \cite{Nagaoka2006-ci}, \cite{Audenaert-et-al-2008}, and \cite{Audenaert2012-mo} the optimality part of the proof of the corresponding bound (Chernoff bound in the case of \cite{Nussbaum-Szkola-2009}, Hoeffding bound in the case of \cite{Nagaoka2006-ci, Audenaert-et-al-2008}, and the mixed error probabilities related to Chernoff and Hoeffding bounds  in the case of \cite{Audenaert2012-mo}) is obtained by making use of the Nussbaum-Szkoła distributions. 

Mosonyi \cite{Mosonyi-2009} defines Nussbaum-Szkoła distributions in the continuous variable setting and uses them as a tool in solving the problem of discriminating gaussian states that have shift invariant and gauge invariant quasifree parts. The hypothesis testing problem of discriminating general gaussian states is stated in the article as an open question.  This is the first instance of direct use of Nussbaum-Szkoła distributions in the infinite dimensional setting. Other recent instances where these distributions are used in the infinite dimensional setting can be seen in the articles \cite{androulakis-john-2022b, Androulakis-John-2023} by the present authors. The  definition of Nussubaum-Szkoła distributions given by Mosonyi agree with Definition \ref{defn:nussbaum-szkola}  in this article, but unlike Mosonyi, we do not restrict ourselves to gaussian states.

Hiai, Mosonyi, Petz and Bény \cite{Hiai2011-qb, Hiai2017-sg} use Nussbaum-Szkola distributions for the first time to study general $f$-divergences between states of finite dimensional $C^*$-algebras. Since they work in finite dimensions
they take a simplified approach in the study of the relative modular operator (see \cite[last line of page 695]{Hiai2011-qb}). Since in the present article, we work with states on finite and infinite dimensional Hilbert spaces, our detailed approach 
to the study of the relative modular operator pays off. Lemma 2.9 in \cite{Hiai2011-qb} is the finite dimensional version
of the main result of our article (Theorem~\ref{thm:f-divergence}). 
The flexibility that we obtain by considering finite and infinite dimensions in our main result is crucial for the applications that we derive  in \cite{androulakis-john-2022b, Androulakis-John-2023}. 

Tomamichel and Hayashi \cite{Tomamichel-Hayashi-2013} use the methods of one-shot entropy, (in particular, Renner's smooth min-entropy \cite{renner-2008}), and the information spectrum in order to compute 
tight second order asymptotics of several entropic quantities.  These entropic quantities are 
important in data compression  and randomness extraction. The Nussbaum-Szkoła distributions play an important role in \cite{Tomamichel-Hayashi-2013}. 
Furthermore, Datta, Mosonyi, Hsieh and Brandão \cite{Datta2013-yg} use smooth entropy approach and obtain second order asymptotics for quantum hypothesis
testing.
Li \cite{Li2014-xg} independently of \cite{Tomamichel-Hayashi-2013}, also uses Nussbaum-Szkoła distributions in order to obtain second order asymptotics for 
entropic quantities which are important in quantum hypothesis testing. Moreover, Nussbaum-Szkoła distributions are useful in the work of Datta, Pautrat, and  Rouzé \cite{Datta2016-lb} that extends the results of 
\cite{Tomamichel-Hayashi-2013, Li2014-xg} to the non-i.i.d. case, (this includes Gibbs states of quantum spin systems and quasi-free states of fermionic lattice gases).

Recall that the classical capacity of a quantum channel is defined as the maximum number of bits that can be transmitted per use of the 
quantum channel in such a way that the decoding error approaches zero as the codeword length $n$ tends to infinity. This classical capacity is described by the Coding Theorems of Holevo \cite{holevo1973bounds,Holevo-1998} and Schumacher-Westmoreland \cite{Schumacher-Westmoreland-1997}, (now known as the HSW Theorem).
A natural question is whether one is able to reliably transmit  classical information via the use of a quantum channel
at a rate asymptotically close to its classical capacity. The Nussbaum-Szkoła distributions have also been 
used in answering this question.  More precisely, consider the problem of transmitting classical information 
via a quantum channel at the rate 
$\text{exp}\left(\Theta(n^{-t})\right)$ bits per channel use close to its capacity, for $t \in [0, 1/2]$. This problem
can be broken into three cases of interest.
The cases $t=0$, $t \in (0,1/2)$ and $t=1/2$ are respectively called \lq\lq large deviation\rq\rq\,, \lq\lq moderate deviation\rq\rq\, and \lq\lq small deviation\rq\rq\   regimes. Two prominent methods have been employed in solving this
question. The first method is by obtaining second order asymptotics for the quantity $\log M^*(W^n,\epsilon)$,
where $M^*(W^n,\epsilon)$ denotes the maximum number of symbols that can be transmitted via $n$ many uses 
of the quantum channel $W$, such that the average probability of error is less than $\epsilon$. The second method 
is by proving tight sphere packing bounds. Hayashi \cite{Hayashi2007-cp} and Dalai \cite{Dalai2013-yl} study the large
deviation regime. Moreover, Dalai \cite{Dalai2013-yl} obtains his result by proving sphere packing
bounds for classical-quantum channels while following a method of Shannon, Gallager, and Berlekamp \cite{Shannon1967-bm}.
 The passage from the classical case of \cite{Shannon1967-bm} to the classical-quantum channels  is achieved via the 
 Nussbaum-Szkoła distributions. 
A related work which uses Dalai's ideas as well as the Nussbaum-Szkoła distributions is done by  Chung, Guha and Zheng \cite{Chung2016-tw} to study superadditivity of quantum channel coding rate. The results of \cite{Dalai2013-yl} have been improved by Cheng, Hsieh and Tomamichel \cite{Cheng2019-hs}.
Chubb, Tan and Tomamichel \cite{Chubb2017-zx} study the moderate deviations regime using the method of the second order asymptotics.
Cheng and Hsieh \cite{Cheng2018-yg} also study the moderate deviations regime, by using the method of sphere packing. Tomamichel and Tan 
\cite{Tomamichel2015-yp} investigate the small deviation regime for certain classical-quantum channels
(e.g.\ image-additive, memoryless channels), using the method of second order asymptotics.

Another instance where Nussbaum-Szkoła distributions have been used in the literature is for coherence distillation. More precisely, Hayashi, Fang and Wang \cite{Hayashi2021-kr} use Nussbaum-Szkoła distributions in order to derive relationships
between entropic quantities of tripartite quantum systems.

The Nussbaum-Szkoła distributions have also been used in the literature to derive useful inequalities between quantum divergences
 from the corresponding inequalities about classical divergences. A particular instance of this situation
can be viewed in the work of Dupuis and Fawzi \cite[Lemma IV.1]{Dupuis2019-vu},  which is further used to improve the second order term in the Entropy Accumulation Theorem. In Section~\ref{sec:quant-f-div-inequalities} of our article we provide many similar examples.

\section{Quantum \texorpdfstring{\ensuremath{f}}{}-divergences via Nussbaum-Szko{\l}a Distributions}\label{sec:f-divergence} 
By a state on a Hilbert space we mean a positive trace class operator with unit trace.  Let $\rho$ and $\sigma$ be any two states on a Hilbert space $\CK$.  The relative modular operator $\Delta_{\rho,\sigma}$ with respect to the states $\rho$ and $\sigma$ was introduced by Araki in \cite{araki1976relative} and \cite{araki1977relative}. It is (in general) an unbounded positive selfadjoint operator defined on a dense subspace of the Hilbert space $\mathcal{B}_2(\mathcal{K})$ of Hilbert-Schmidt operators on $\mathcal{K}$, where the scalar product on $\mathcal{B}_2(\mathcal{K})$ is denoted as $\braket{\cdot}{\cdot}_2$.
 Let $\xi^{\Delta_{\rho,\sigma}}$ denote the spectral measure associated with $\Delta_{\rho,\sigma}$. A detailed analysis of the relative modular operator in our setting, and its spectral decomposition is provided in Appendix \ref{appendix:relative-modular}. Keeping these notations we define the quantum $f$-divergence of $\rho,\sigma$ in Definition \ref{defn:f-divergence}. Our definition of quantum $f$-divergences is same as that of Hiai in \cite[Definition 2.1]{Hiai-2018} adapted to the specific von~Neumann algebra $\B{\CH}$. 

Before defining the $f$-divergence, we need to fix a few notations and conventions related to a convex (or concave) function 
 $f:(0,\infty)\to \BR$. First of all let
 \begin{align}\label{eq:convex-f-notations-1}
     f(0) := \lim_{t\downarrow 0}f(t),&\quad f'(\infty):= \lim_{t\rightarrow\infty}\frac{f(t)}{t}.
 \end{align}
 Moreover, we use the conventions,
 \begin{align}\label{eq:convex-f-notations-2}
   \begin{split}
      & 0f\left(\frac{0}{0}\right) = 0, \\
    & 0f\left(\frac{a}{0}\right) :=  \lim\limits_{t\rightarrow 0}tf\left(\frac{a}{t}\right) = a\lim\limits_{s\rightarrow \infty}\frac{f(s)}{s}=af'(\infty),  \text{ for } a>0,\\
     & 0\cdot (\pm\infty)= 0
   \end{split}
 \end{align} 
 along with 
 $a \cdot (\pm\infty) = \pm\infty$, for $a >0$.
 In what follows,  the notation ``$\int\limits_{0^+}^\infty $'' is used to denote an integral over the open interval $(0,\infty)\subseteq \BR$. For a state $\eta$, let $\Pi_{\eta}$ denote the orthogonal projection onto the support of $\eta$ and $\Pi_{\eta}^\perp$, the projection onto $\ker \eta$.
 
 Now we define the $f$-divergence of two states $\rho$ and $\sigma$ as in \cite[Definition 2.1]{Hiai-2018}. A motivation for this definition can be seen in \cite[Equations 3.9 and 3.12]{Hiai-Mosonyi-2017}.
 \begin{defn}\label{defn:f-divergence}
 Let $\rho$ and $\sigma$ be states on a Hilbert space ${\CK}$. If $f:(0,\infty)\to\BR$ is a convex (or concave) function then the $f$-divergence $D_f(\rho||\sigma)$ of $\rho$ from $\sigma$ is defined as 
 \begin{align}\label{eq:f-divergence}
     D_f(\rho||\sigma)= \int_{0^+}^{\infty}f(\lambda)\mel{\sqrt{\sigma}}{\xi^{\Delta_{\rho,\sigma}}(\dd \lambda)}{\sqrt{\sigma}}_2+f(0)\tr \left(\sigma\Pi_\rho^\perp\right)+f'(\infty)\tr\left(\rho\Pi_\sigma^\perp\right).
 \end{align}
 \end{defn}
The reader is warned that we use the same notation $D_f$ for both classical and quantum $f$-divergence. It will be clear from the context whether we use the quantum or classical divergence each time.
\begin{rmk}\label{rmk:f-divergence}
It is known that, if $f$ is convex then $D_f(\rho||\sigma)$ is well defined and takes value in $(-\infty,\infty]$ for every state $\rho$ and $\sigma$. A proof of this fact can be seen in \cite[Lemma 2.1]{Hiai-2018}. In particular, we have
\begin{align}\label{eq:f-divergence-integral-finite-convex}
\int_{0^+}^{\infty}f(\lambda)\mel{\sqrt{\sigma}}{\xi^{\Delta_{\rho,\sigma}}(\dd \lambda)}{\sqrt{\sigma}}_2>-\infty, \quad \left(\text{when }f \text{ is convex}\right), 
\end{align} since the second term and the third term in \eqref{eq:f-divergence} are strictly bigger than $-\infty$.
If $f$ is concave, then $-f$ is convex and we get \begin{align}\label{eq:f-divergence-integral-finite-concave}
    \int_{0^+}^{\infty}f(\lambda)\mel{\sqrt{\sigma}}{\xi^{\Delta_{\rho,\sigma}}(\dd \lambda)}{\sqrt{\sigma}}_2<\infty, \quad \left(\text{when }f \text{ is concave}\right).
\end{align}
Furthermore, if $f$ is convex, then  neither $f(0)$ nor $f'(\infty)$ can be equal to $-\infty$. Similarly, if $f$ is concave, then neither  $f(0)$ nor $f'(\infty)$ can be equal to $\infty$. This, along with Equations~\eqref{eq:f-divergence-integral-finite-convex} and \eqref{eq:f-divergence-integral-finite-concave}  ensures that 
if $f$ is convex or concave, then $D_f(\rho||\sigma)$ is well defined. Also, notice that without the assumption of convexity (concavity), it is possible that $f(0)$ and $f'(\infty)$ can be infinities with opposite sign, in which case the last two terms in the Equation~\eqref{eq:f-divergence} cannot be added. 
\end{rmk}
\begin{egs}\label{eg:f-div} We can create various examples of  relative entropic quantities by taking different functions $f$ in the definition of $f$-divergence.  The following are important examples motivated from the corresponding classical counterparts \cite{sason-verdu-2016}.
    \begin{enumerate}
    \item $f(t)=t \log t \Rightarrow D(\rho\|\sigma):=D_f(\rho \| \sigma)$, the \textbf{Umegaki Relative Entropy}. 
    \item For $\alpha\in (0,1)\cup (1,\infty)$, $f_\alpha(t)=t^\alpha $ we obtain   $ D_{\alpha}(\rho||\sigma) := \frac{1}{\alpha-1} \log D_{f_\alpha}(\rho\|\sigma)$, the  \textbf{Petz-R\'enyi} $\bm{\alpha}$-\textbf{relative entropy}. 
    \item For $\alpha\in (0,1)\cup (1,\infty)$, $f_{\alpha}(t)=\frac{t^\alpha-1}{\alpha-1} 
$ we obtain   $ \mathscr{H}_{\alpha}(\rho||\sigma)=D_f(\rho \| \sigma)$, the \textbf{Quantum Hellinger $\alpha$-divergence}. 
\item  $f(t)=|t-1| \Rightarrow V(\rho||\sigma):= D_f(\rho \| \sigma)$, the \textbf{Quantum Total Variation}. It may be noted that, in the literature the terminology quantum total variation  is usually reserved for the quantity $\frac{1}{2}\norm{\rho-\sigma}_1$, where $\norm{\cdot}_1$ denotes the trace distance. In this article,  we deviate from this terminology because using the same function $f$ the corresponding classical $f$-divergence is  called the total variation. Moreover, we provide applications for this quantity in Section \ref{sec:quant-f-div-inequalities} in equations \eqref{item:1.1}, \eqref{item:5.1}, \eqref{item:7.1}, \eqref{item:8.1}, \eqref{item:8.2}, and \eqref{item:10.1}, where this quantity appears naturally in several interesting bounds among different quantum $f$-divergences.
\item  $f(t)=(t-1)^2 \Rightarrow \chi^2(\rho||\sigma):= D_f(\rho \| \sigma)$, the \textbf{Quantum }$\bm{\chi^2}$-\textbf{divergence}. 
\end{enumerate}
\end{egs}
Now we define the Nussbaum-Szko{\l}a distributions associated with the states $\rho$ and $\sigma$, (see \cite{Nussbaum-Szkola-2009} for the original definition in the finite dimensional setting, and \cite{Mosonyi-2009} for an infinite dimensional version in the setting of gaussian states). Our definition is valid in both finite and infinite dimensions.
 \begin{defn}\label{defn:nussbaum-szkola}
  (\textit{Nussbaum-Szko{\l}a distributions.})
  Let $\CK$ be a complex Hilbert space with  $\dim \CK = \abs{\mathcal{I}}$, where $\mathcal{I} =\{1,2,\dots,n\}$ for some $n\in \BN$, or $\mathcal{I}=\BN$. Let $\rho$ and $\sigma$ be states on $\CK$  with spectral decomposition \begin{align}\label{eq:spectral-rho-sigma}\begin{split}
    \rho &= \sum_{i\in \mathcal{I}}r_i \ketbra{u_i}, \quad r_i\geq 0,\quad \sum_{i\in \mathcal{I}} r_i = 1,\quad \{u_i\}_{i\in \mathcal{I}} \text{ is an orthonormal basis of }  \mathcal{K};\\
    \sigma &= \sum_{j\in \mathcal{I}}s_j \ketbra{v_j}, \quad s_j\geq0, \quad \sum_{j\in \mathcal{I}} s_j = 1, \quad \{v_j\}_{j\in \mathcal{I}}\text{ is an orthonormal basis of } \mathcal{K}.
    \end{split}
\end{align}
 Define the Nussbaum-Szko{\l}a distribution $P$ and $Q$ associated with $\rho$ and $\sigma$ on $\mathcal{I}\times \mathcal{I}$ by,
\begin{align}
    \label{eq:P-and-Q}\begin{split}
  P(i,j) &= r_i\abs{\braket{u_i}{v_j}}^2,\\
  Q(i,j)& = s_j\abs{\braket{u_i}{v_j}}^2, \quad \forall (i,j)\in \mathcal{I}\times\mathcal{I}.
     \end{split}
\end{align}
 \end{defn}
  \begin{rmk}Observe first that
 \[\sum_{i,j}P(i,j) = \sum_{i}r_i\sum_j\abs{\braket{u_i}{v_j}}^2 = \sum_{i}r_i = 1,\]
  since $\{u_i\}_{i}$ and $\{v_j\}_{j}$ are orthonormal bases. Similarly, $\sum_{i,j}Q(i,j) = 1$. Hence $P$ and $Q$ are probability distributions on $\mathcal{I}\times \mathcal{I}$.  
 \end{rmk}
 For any two probability distributions and for any
 convex (or concave) function $f$ on $(0, \infty)$ the $f$-divergence of one distribution from the other is defined
 as in Definition~\ref{defn:classical-f-divergence} and the discrete special case is described in Remark \ref{defn:renyi-divergence-Kullback}.
 Under the same notations and conventions as in \eqref{eq:convex-f-notations-1} and \eqref{eq:convex-f-notations-2} the classical $f$-divergence  $D_f(P||Q)$ of $P$ from $Q$ for a convex (or concave) $f$ on $(0,\infty)$ is 
 \begin{equation}\label{eq:f-divergence-Nussbaum-Szkola}
  D_f(P||Q) = \sum_{i,j\in \mathcal{I}}f\left(\frac{P(i,j)}{Q(i,j)}\right)Q(i,j).
 \end{equation}
 It is known that the sum of the series defining $D_f(P||Q)$ as above is independent of any rearrangement of the series but we provide a proof of this fact in Remark \ref{defn:renyi-divergence-Kullback} in the Appendix.  We refer to \cite{Liese-Vajda-2006,Csiszar-Sheilds-2004} for various properties of classical $f$-divergences. 
 In the next lemma, we compute a useful formula for the classical $f$-divergence of the Nussbaum-Szko\l a distributions.

\begin{lem}\label{lem:f-divergence-Nussbaum-Szkola-formula} The $f$-divergence of the Nussbaum-Szko{\l}a distributions can be computed as
 \begin{equation}\label{eq:f-divergence-Nussbaum-Szkola-formula}
     D_f(P||Q) 
     = \sum\limits_{\left\{\substack{i,j\,:\\ r_i s_j \neq 0} \right\}}f\left({r_i}{s_j^{-1}}\right)s_j\abs{\braket{u_i}{v_j}}^2+f(0)Q(P=0)+f'(\infty)P(Q=0).
 \end{equation}
\end{lem} 
\begin{proof} 
By \eqref{eq:df-pq-expression} in appendix,  we have
\begin{align*}
    D_f(P||Q) 
    &=\sum\limits_{\left\{\substack{i,j\,: r_i\neq 0, s_j \neq 0,\\ \braket{u_i}{v_j}\neq 0} \right\}}f\left(\frac{r_i}{s_j}\right)s_j\abs{\braket{u_i}{v_j}}^2+\sum\limits_{\left\{\substack{i,j\,: r_i= 0, s_j \neq 0,\\\braket{u_i}{v_j}\neq 0} \right\}}f(0)s_j\abs{\braket{u_i}{v_j}}^2\\&\phantom{.......................................................}+\sum\limits_{\left\{\substack{i,j\,: r_i\neq 0, s_j =0,\\\braket{u_i}{v_j}\neq 0} \right\}}f'(\infty)r_i\abs{\braket{u_i}{v_j}}^2,
\end{align*}
which is same as \eqref{eq:f-divergence-Nussbaum-Szkola-formula}. 
\end{proof}

Now we analyse the measure involved in the definition of quantum $f$-divergence. Before stating the next lemma we introduce the following convention which will be used in its statement and proof: 

When a sum has an empty index set, then the sum is equal to zero (either zero scalar or zero vector of the appropriate vector space depending on the context), i.e.,
\[\sum_{\emptyset}(\cdot) = 0.\]
\begin{lem}\label{lem:measure-f-divergence}The measure $E\mapsto\mel{\sqrt{\sigma}}{\xi^{\Delta_{\rho,\sigma}}(E)}{\sqrt{\sigma}}_2$ on the Borel sigma algebra of $\BR$ is supported on the set $\{0\}\cup\{r_is_j^{-1}\,:\,r_i\neq0, s_j\neq0\}$. For $\lambda \neq 0$,
 \begin{align}\label{eq:measure}
     \mel{\sqrt{\sigma}}{\xi^{\Delta_{\rho,\sigma}}\{ \lambda \}}{\sqrt{\sigma}}_2= \sum\limits_{\{i,j\,:\, r_{i}s_{j}^{-1} =\lambda \}}s_{j}\abs{\braket{u_{i}}{v_{j}}}^{2},
 \end{align}
and 
 \begin{align}\label{eq:measure-at-zero}
     \mel{\sqrt{\sigma}}{\xi^{\Delta_{\rho,\sigma}}\{0\}}{\sqrt{\sigma}}_2= 
     \sum\limits_{\left
             \{\substack{i,j\,:\, r_{i} =0, s_j\neq 0\\\braket{u_i}{v_j}\neq 0} \right\}}s_{j}\abs{\braket{u_{i}}{v_{j}}}^{2}=Q(P=0).
 \end{align}
\end{lem}
\begin{proof}
Recall from Proposition \ref{prop:spectral-relative-modular} that  the spectrum of the relative modular operator $\Delta_{\rho,\sigma}$ is supported on the set $\{0\}\cup\{r_is_j^{-1}\,:\,r_i\neq0, s_j\neq0\}$. Hence it is clear that the measure in the statement of the lemma is supported on the same set.
 By the spectral decomposition of the relative modular operator $\Delta_{\rho,\sigma}$ obtained in Proposition \ref{prop:spectral-relative-modular}, the spectral measure $\xi^{\Delta_{\rho,\sigma}}$ is supported on the eigenvalues of $\Delta_{\rho,\sigma}$, 
 and satisfies

 \begin{align*}
    \xi^{\Delta_{\rho,\sigma}}\{\lambda\} =\begin{cases}\sum\limits_{\{i,j\,:\, r_{i}s_{j}^{-1} =\lambda \}} \ketbra{X_{ij}}, & \text{for } \lambda \neq 0 ;\\
    &\\
   \hfill \sum\limits_{\{i,j\,:\,r_i=0 \text{ or }s_j = 0\}} \ketbra{X_{ij}}, & \text{for } \lambda =0,
    \end{cases}
 \end{align*}
where 
    $X_{ij} = \ketbra{u_i}{v_j} \in \Bt{\CK},$ for all $i,j\in \mathcal{I}$. Therefore for $\lambda\neq 0$ we have,
 \begin{align*}
     \mel{\sqrt{\sigma}}{\xi^{\Delta_{\rho,\sigma}}\{\lambda\}}{\sqrt{\sigma}}_2 &=\sum\limits_{\{i,j\,:\, r_{i}s_{j}^{-1} =\lambda\}} \braket{\sqrt{\sigma}}{X_{ij}}_2\braket{X_{ij}}{\sqrt{\sigma}}_2\\
     &= \sum\limits_{\{i,j\,:\, r_{i}s_{j}^{-1} =\lambda\}}\abs{\braket{\sqrt{\sigma}}{X_{ij}}_2}^2\\
     &= \sum\limits_{\{i,j\,:\, r_{i}s_{j}^{-1} =\lambda\}}\abs{\tr \left\{X_{ij}\sqrt{\sigma}\right\}}^2\\
     &= \sum\limits_{\{i,j\,:\, r_{i}s_{j}^{-1} =\lambda\}}\abs{\tr \left\{\left(\ketbra{u_{i}}{v_{j}}\right)\left(\sum\limits_{k}\sqrt{s_k}\ketbra{v_k}\right)\right\}}^2\\
      &= \sum\limits_{\{i,j\,:\, r_{i}s_{j}^{-1} =\lambda\}}s_{j}\abs{\braket{u_{i}}{v_{j}}}^{2}.
 \end{align*}
 A similar computation as above shows that 
 \begin{align*}
   \mel{\sqrt{\sigma}}{\xi^{\Delta_{\rho,\sigma}}\{0\}}{\sqrt{\sigma}}_2 &=\sum\limits_{\{i,j\,:\, r_{i}=0 \text{ or }s_{j}=0\}} \braket{\sqrt{\sigma}}{X_{ij}}_2\braket{X_{ij}}{\sqrt{\sigma}}_2\\  
   &=\sum\limits_{\{i,j\,:\, r_{i}=0 \text{ or }s_{j}=0\}}s_j\abs{\braket{u_i}{v_j}}^2\\
   &=\sum\limits_{\left\{\substack{i,j\,:\, r_{i}=0, s_{j}\neq0\\\braket{u_i}{v_j}\neq 0}\right\}}s_j\abs{\braket{u_i}{v_j}}^2\\
   &=Q(P=0).
 \end{align*}
\end{proof}

In several occasions below, we will use the following rearrangement trick for a sum of the form $\sum_{k}f(x_k)y_k$ with $y_k>0$ for all $k$.  Notice that if the sum of the negative terms in the series above is strictly bigger than $-\infty$ (or the sum of  positive terms in the series is strictly less than $\infty$), then any rearrangement of the series produces the same sum. In  particular, for $N \in \mathbb{N} \cup \{ \infty \}$, if the sum of the negative terms in the series $ \sum\limits_{k=1}^N f(x_k)y_k$ is strictly bigger than $-\infty$, (or the sum of  its positive terms is strictly less than $\infty$), we have
\begin{equation}\label{eq:trivial-identity}
     \sum\limits_{k=1}^N f(x_k)y_k =\sum\limits_{\lambda \in \{ x_k: k=1, \ldots , N \} }f(\lambda)
     \sum\limits_{\{\ell\,:\,x_{\ell}=\lambda\}}y_{\ell}. 
 \end{equation}
 Note that in the first sum on the right side of \eqref{eq:trivial-identity}, every element $x$ in the sequence $(x_k)_{k=1}^N$ appears exactly once even if the terms $x_k$ are not distinct.
Our next theorem is  the main result in this article. It proves that the quantum $f$-divergence of two states is same as the classical $f$-divergence of corresponding Nussbaum-Szko{\l}a distributions.
\begin{thm}\label{thm:f-divergence}
Let $\rho,\sigma$ be as in (\ref{eq:spectral-rho-sigma}) and $P,Q$ denote the corresponding Nussbaum-Szko{\l}a distributions.  Let $f:(0,\infty)\to\BR$ be a convex (or concave) function  and $D_f(\rho||\sigma)$, $D_{f}(P||Q)$  respectively   denote the quantum $f$-divergence of $\rho$ from $\sigma$ and the classical $f$-divergence  of  $P$ from $Q$. Then \begin{align}
    \label{eq:f-divergence-quantum-equals-classical}
    D_f(\rho||\sigma) = D_f(P||Q).
\end{align}
\end{thm}
\begin{proof} We will show that each term on the right side of \eqref{eq:f-divergence} matches with the corresponding terms on the right side of \eqref{eq:f-divergence-Nussbaum-Szkola-formula}
We first evaluate the first term on the right side of \eqref{eq:f-divergence}. By Lemma \ref{lem:measure-f-divergence} we have 
   \begin{align*}
\int_{0^+}^{\infty}f(\lambda)\mel{\sqrt{\sigma}}{\xi^{\Delta_{\rho,\sigma}}(\dd \lambda)}{\sqrt{\sigma}}_2&= \sum_{\lambda\in \text{sp}(\Delta_{\rho,\sigma})\setminus \{0\}}f(\lambda)\sum\limits_{\{i,j\,:\, r_{i}s_{j}^{-1} =\lambda\}}s_{j}\abs{\braket{u_{i}}{v_{j}}}^{2}.
   \end{align*}
   If the function $f$ is convex  then by Equation \eqref{eq:f-divergence-integral-finite-convex} the above sum is strictly bigger than $-\infty$ and hence the sum of its negative terms is also bigger than $-\infty$. Thus \eqref{eq:trivial-identity} is applicable. Similarly, if $f$ is concave one can use \eqref{eq:f-divergence-integral-finite-concave} and conclude that \eqref{eq:trivial-identity} is applicable in this case as well. Therefore, we have
\[ \int_{0^+}^{\infty}f(\lambda)\mel{\sqrt{\sigma}}{\xi^{\Delta_{\rho,\sigma}}(\dd \lambda)}{\sqrt{\sigma}}_2 =\sum_{\left\{\substack{i,j\,:\,\\r_i\neq0, s_j\neq 0}\right\}}f(r_is_j^{-1})s_j\abs{\braket{u_i}{v_j}}^2.\]
   This is same as the first term on the right side of  \eqref{eq:f-divergence-Nussbaum-Szkola-formula}. Now \begin{align*}
       \tr \left(\sigma\Pi_\rho^\perp\right)&= \tr \left(\sum_js_j\ketbra{v_j}{v_j}\sum_{\{i\,:\,r_i=0\}}\ketbra{u_i}{u_i}\right)\\
       &=\sum_{\left\{\substack{i,j\,:\,\\r_i=0}\right\}}s_j\abs{\braket{u_i}{v_j}^2}.
   \end{align*}
 Therefore, \[f(0) \tr \left(\sigma\Pi_\rho^\perp\right) = f(0)Q(P=0),\] which is same as the second term  on the right side  of \eqref{eq:f-divergence-Nussbaum-Szkola-formula}. A similar computation as above shows that
   \[f'(\infty)\tr\left(\rho\Pi_\sigma^\perp\right) = f'(\infty)P(Q=0),\] which is same as the third term  on the right side of  \eqref{eq:f-divergence-Nussbaum-Szkola-formula}. Thus we proved \eqref{eq:f-divergence-quantum-equals-classical}.
\end{proof}

The correspondence between quantum states and Nussbaum-Szkoła distributions we defined in this article have the following important properties.

\begin{prop}\label{prop:rho-equals-sigma-P-equals-Q}
Let $\rho$ and $\sigma$ be as in (\ref{eq:spectral-rho-sigma})  and let $P$ and $Q$ be as in (\ref{eq:P-and-Q}) then 
\[P=Q \Leftrightarrow \rho=\sigma.\]
\end{prop}
The formula for classical  $f$-divergence simplifies if $P\ll Q$ (see \eqref{eq:f-div}). This property for the Nussbaum-Szkoła distributions translate to the condition $\supp \rho \subseteq \supp \sigma$ as shown in the next proposition.
 \begin{prop}\label{prop:support-condition-iff-absolute-continuity}
Let $\rho$ and $\sigma$ be as in (\ref{eq:spectral-rho-sigma})  and let $P$ and $Q$ be as in (\ref{eq:P-and-Q}), then \[\supp \rho \subseteq \supp \sigma \Leftrightarrow P\ll Q.\]
\end{prop}
For proving Proposition \ref{prop:rho-equals-sigma-P-equals-Q}, we need the following lemma.
\begin{lem}\label{lem:support-condition}
Let $\rho$ and $\sigma$ be as in (\ref{eq:spectral-rho-sigma}). Then  $\supp \rho \subseteq \supp \sigma$ if and only if $s_j=0$ for some $j$ implies that for every $i$ at least one of the two quantities $\{\braket{u_i}{v_j}, r_i \}$ is equal to zero.  
\end{lem}
\begin{proof} ($\Rightarrow$) Assume $\supp \rho \subseteq \supp \sigma$.
The assumption that $s_j=0$ implies that $v_j\in \ker \sigma$. Now the condition $\braket{u_i}{v_j}\neq 0$ implies that $u_i\notin (\ker{\sigma})^{\perp} =\supp \sigma$. Since $\supp \sigma\supseteq\supp \rho$ and $u_i$ is an eigenvector of $\rho$ we see that \[u_i\notin\supp\rho \Rightarrow u_i\in \ker \rho \Rightarrow r_i=0.\]

($\Leftarrow$) We will prove that $\ker \sigma \subseteq \ker \rho$. It is enough to prove that every $v_j$ for which $s_j = 0$ belong to $\ker \rho$. Fix such  $j$ such that $s_j=0$, by our assumption, $\braket{u_i}{v_j}=0$ for all $i$ such that $r_i\neq 0$ and hence $\sum_{\{i:r_i\neq 0\}}\braket{u_i}{v_j}\ket{u_i} = 0$. This means that the projection of $v_j$ to  $\supp \rho$ is  the zero vector which means that $v_j\in \ker \rho$.
\end{proof}

\subsubsection*{Proof of Proposition \ref{prop:rho-equals-sigma-P-equals-Q}}
\begin{proof}  Clearly, if $\rho=\sigma$ then $P=Q$. Now assume $P=Q$.
By definition \begin{align}\label{eq:P-equals-Q-implies}
    P=Q \Rightarrow r_i\abs{\braket{u_i}{v_j}}^2 = s_j\abs{\braket{u_i}{v_j}}^2,\quad \forall i,j.
\end{align}
Therefore, $s_j=0$ for some $j$ implies that for every $i$ at least one of the two quantities $\{\braket{u_i}{v_j}, r_i \}$ is equal to zero. Now by Lemma \ref{lem:support-condition}, $\supp\rho \subseteq \supp \sigma$. Also $\supp \sigma \subseteq \supp \rho$ by symmetry of the situation. Hence we have \[\supp\rho = \supp \sigma.\]
Thus $\ker\rho = \ker\sigma$. Therefore, to prove that $\rho = \sigma$, we will show that the nonzero eigenvalues and the corresponding eigenspaces of $\rho$ and $\sigma$ are same.
Now fix $i_0\in \mathcal{I}$ such that $r_{i_0}\neq 0$. If $r_{i_0}\neq s_j$ for all $j$, the second equality in (\ref{eq:P-equals-Q-implies}) shows that $\braket{u_{i_0}}{v_j} =0$ for all $j$ such that $s_j\neq 0$ but this is impossible because $\{v_j|s_j\neq0\}$ is an orthonormal basis for $\supp \sigma$ and $0\neq u_{i_0}\in \supp \sigma$. Therefore, for each $i_0\in \mathcal{I}$ such that $r_{i_0}\neq 0$, there exists $j_0\in \mathcal{I}$ such that $r_{i_0}=s_{j_0}.$ Hence $\operatorname{sp}(\rho) \subseteq\operatorname{sp}(\sigma)$, where `$\operatorname{sp}$' denotes spectrum. A similar argument shows that $\operatorname{sp}(\sigma) \subseteq \operatorname{sp}(\rho)$. Thus we have  \[\operatorname{sp}(\rho) =\operatorname{sp}(\sigma).\]
Fix ${i_0}$ and ${j_0}$ such that $r_{i_0}=s_{j_0}\neq 0$.
Let $R_{i_0}$ denote the eigenspace of $\rho$ corresponding to the eigenvalue $r_{i_0}$ and $S_{j_0}$ denote the eigenspace of $\sigma$ corresponding to the eigenvalue $s_{j_0}$. We will show that $R_{i_0}=S_{j_0}$, which will complete the proof since $r_{i_{0}}$ is an arbitrary non zero eigenvalue. It is enough to show the following two claims: \begin{enumerate}
    \item \label{claim:1} If $r_i=r_{i_0}$ for some $i$ then $u_i\perp v_j$ for all $j\in \mathcal{I}$ with $s_j\neq s_{j_0}$;
    \item If $s_j=s_{j_0}$ for some $j$ then $v_j\perp u_i$ for all $i\in \mathcal{I}$ with $r_i\neq r_{i_0}$.
\end{enumerate}
Proof of both the claims above are similar so we will prove \ref{claim:1} only. Fix $i$ such that $r_i=r_{i_0}$ and $j$ such that $s_j\neq s_{j_0}$. Since $r_{i_0}=s_{j_0}$, we have $r_i\neq s_j$. Thus by \eqref{eq:P-equals-Q-implies} we have $\braket{u_i}{v_j}=0$, which proves \ref{claim:1}.
\end{proof}
\subsubsection*{Proof of Proposition \ref{prop:support-condition-iff-absolute-continuity}}
\begin{proof}Assume $ P\nll Q$, we have
\begin{align*}
    P\nll Q &\Leftrightarrow \exists (i,j) \in \mathcal{I}\times \mathcal{I} \textnormal{ such that } \braket{u_i}{v_j} \neq 0, s_j=0,r_i\neq 0\\
    &\Leftrightarrow \exists (i,j) \in \mathcal{I}\times \mathcal{I} \textnormal{ such that } \braket{u_i}{v_j} \neq 0, v_j\in \ker \sigma, u_i \in \supp \rho\\
    &\Leftrightarrow \exists i\in \mathcal{I} \textnormal{ such that } u_i \in \supp \rho,  u_i \notin (\ker \sigma )^\perp\\
    &\Leftrightarrow \exists i\in \mathcal{I} \textnormal{ such that } u_i \in \supp \rho,  u_i \notin \supp \sigma\\
     &\Leftrightarrow \supp \rho \nsubseteq \supp \sigma.
\end{align*}
\end{proof}

\section{Quantum \texorpdfstring{\ensuremath{f}}{}-divergence Inequalities}\label{sec:quant-f-div-inequalities}
Our main Theorem~\ref{thm:f-divergence} provides a framework to obtain quantum results from classical ones. To illustrate this, we
 list a few quantum $f$-divergence inequalities that follow immediately from the corresponding classical counterparts.  
 The notations used below are as in Example \ref{eg:f-div}. Note also, that we use the same notations to denote any specific classical or quantum $f$-divergence. It will be clear from the context whether we are discussing quantum or classical case (we reserve the letters $\rho$ and $\sigma$ to denote quantum states and the letters $P$ and $Q$  for classical probability distributions). The list of inequalities provided below are not exhaustive by any means; our purpose is to show the use of our main result in obtaining such inequalities. One may refer to the existing literature on classical $f$-divergences for deducing several other quantum $f$-divergence results from the existing classical ones. For example,  the article \cite{sason-verdu-2016} contains several $f$-divergence inequalities that are not presented here but are easily generalized to the quantum case using our main theorem in this article.

\begin{enumerate}
    \item The squared Hellinger distance in classical probability  (Example \ref{eg:classical-examples}, item \ref{item:hellinger}) satisfies the following bounds with the total variation distance \label{item:1}
 \cite[p. 25]{Lecam-grace1990}
$$
\begin{aligned}
\mathscr{H}^2(P \| Q) & \leq V(P\|Q)^2 \leq \mathscr{H}(P \| Q)\sqrt{2-\mathscr{H}^2(P \| Q)},
\end{aligned}
$$ hence we have
\begin{align}\label{item:1.1}
  \mathscr{H}^2(\rho \| \sigma) & \leq V(\rho\|\sigma)^2 
 \leq \mathscr{H}(\rho \| \sigma)\sqrt{2-\mathscr{H}^2(\rho \| \sigma)}.
    \end{align}
\item In the classical case the Kullback-Leibler divergence is bounded above by a function of $\chi^2$-divergence as follows \label{item:3} \cite[Theorem 5]{gibbs-su-2002},
$$
D(P \| Q) \leq \log \left(1+\chi^2(P \| Q)\right),
$$ 
where the same logarithm is used as in the definition of Kullback-Leibler divergence. Hence we have the same bound in the quantum case as well\begin{align} \label{item:3.1}
    D(\rho \| \sigma) &\leq \log \left(1+\chi^2(\rho \| \sigma)\right).
\end{align}
\item \label{item:4}  The classical $\chi^2$-divergence is bounded above by a function of Hellinger $\alpha$-divergence for all $\alpha > 2$ \cite[Corollary 5.6]{Guntuboyina-Saha-schiebinger2013}
$$
\chi^2(P \| Q) \leq\left(1+(\alpha-1) \mathscr{H}_\alpha(P \| Q)\right)^{\frac{1}{\alpha-1}}-1,
$$
Therefore, for all $\alpha > 2$,   \begin{align}\label{item:4.1}
     \chi^2(\rho \| \sigma) \leq\left(1+(\alpha-1) \mathscr{H}_\alpha(\rho \| \sigma)\right)^{\frac{1}{\alpha-1}}-1.
 \end{align}
\item \label{item:5} The $\chi^2$-divergence is bounded below a function of total variation distance as follows
 \cite[eq. (58)]{reid-11a}
$$
\chi^2(P \| Q) \geq\left\{\begin{array}{cl}
V(P\|Q)^2, & V(P\|Q) \in[0,1) \\
\frac{V(P\|Q)}{2-V(P\|Q)}, & V(P\|Q) \in [1,2]
\end{array}\right., 
$$ hence we have the same bound in the quantum case as well
\begin{align}
     \label{item:5.1}
\chi^2(\rho \| \sigma) &\geq\left\{\begin{array}{cl}
V(\rho\|\sigma)^2, & V(\rho\|\sigma) \in[0,1] \\
\frac{V(\rho\|\sigma)}{2-V(\rho\|\sigma)}, & V(\rho\|\sigma) \in(1,2) \end{array}\right..
\end{align}
\item \label{item:6} From the article \cite{Simic2008} we obtain the following inequalities: 
$$
\begin{aligned}
\frac{D^2(P \| Q)}{D(Q \| P)} & \leq \frac{1}{2} \chi^2(P \| Q) ,\quad  \text{\cite[(12) and Remark 4]{Simic2008}},\\
16 \mathscr{H}^4(P \| Q)   &\leq {D(P \| Q) D(Q \| P)} 
 \leq \frac{1}{4} {\chi^2(P \| Q) \chi^2(Q \| P)},\\ &\phantom{...........................}\text{\cite[Proposition 3 (i) and Remark 4]{Simic2008}},\\
8 \mathscr{H}^2(P \| Q)  & \leq D(P \| Q)+D(Q \| P) 
 \leq \frac{1}{2}\left(\chi^2(P \| Q)+\chi^2(Q \| P)\right),\\ &\phantom{...........................} \text{\cite[Proposition 3 (ii) and Remark 4]{Simic2008}},
\end{aligned}
$$
where the constants $16$ and $4$ that appear above are different than \cite{Simic2008} because our definition of $\mathscr{H}^2$ is half of that in \cite{Simic2008}. Moreover, the logarithm in the definition of Kullback-Leibler is taken with base $e$. Therefore, for the quantum case we get
 \begin{align}
\label{item:6.1} 
\frac{D^2(\rho \| \sigma)}{D(\sigma \| \rho)} & \leq \frac{1}{2} \chi^2(\rho \| \sigma),  \\
\begin{split}
16 \mathscr{H}^4(\rho \| \sigma)  & \leq D(\rho \| \sigma) D(\sigma \| \rho)
 \leq \frac{1}{4} \chi^2(\rho \| \sigma) \chi^2(\sigma \| \rho), \end{split}\\
\begin{split}
8 \mathscr{H}^2(\rho \| \sigma)  & \leq D(\rho \| \sigma)+D(\sigma \| \rho)  \leq \frac{1}{2}\left(\chi^2(\rho \| \sigma)+\chi^2(\sigma \| \rho)\right), \end{split}
\end{align} where the logarithm in the definition of Umegaki relative entropy is taken with base $e$.
\item \label{item:7} If the logarithm in the definition of Kullback-Leibler is taken with base $e$,  we have \cite[eq. (2.8)]{diaconis-saloff-1996} 
$$
D(P \| Q) \leq \frac{1}{2}\left(V(P,Q)+\chi^2(P \| Q)\right).
$$(Notice that our definition of total variation is twice the one that appears in \cite{diaconis-saloff-1996}.) With the convention about the logarithm in Umegaki relative entropy, we have the following inequality in the quantum case
\begin{align}
    \label{item:7.1} D(\rho \| \sigma) &\leq \frac{1}{2}\left(V(\rho\|\sigma)+\chi^2(\rho \| \sigma)\right).
\end{align}
\item \label{item:8} The symmetrized versions of the Kullback-Leibler and $\chi^2$-divergences satisfy the following bounds  \cite[Corollary 32]{reid-11a}, 
$$
\begin{aligned}
D(P \| Q)  +D(Q \| P) 
&\geq  2V(P||Q) \log \left(\frac{2+V(P\|Q)}{2-V(P\|Q)}\right), \\
\chi^2(P \| Q)  +\chi^2(Q \| P) 
&\geq  \frac{8V(P\|Q)^2}{4-V(P\|Q)^2}.
\end{aligned} 
$$ Therefore, the corresponding quantum versions also satisfy the same bounds
\begin{align}
    \label{item:8.1}
D(\rho \| \sigma)  +D(\sigma \| \rho) 
&\geq  2V(\rho\|\sigma) \log \left(\frac{2+V(\rho\|\sigma)}{2-V(\rho\|\sigma)}\right), \\
\chi^2(\rho \| \sigma)  +\chi^2(\sigma \| \rho) 
&\geq  \frac{8V(\rho\|\sigma)^2}{4-V(\rho\|\sigma)^2}.\label{item:8.2}
\end{align}
\item\label{item:9} 
The Hellinger $\mathscr{H}_{\alpha}$,  the Rényi $\alpha$-relative entropy $D_\alpha$ and the Kullback-Leibler relative entropy $D$ satisfy the following inequality
\cite[ Proposition 2.15]{liese-vajda-1987}: 
$$
\mathscr{H}_\alpha(P \| Q)  \leq D_\alpha(P \| Q) \leq D(P \| Q)\leq D_{\beta}(P \| Q)\leq \mathscr{H}_{\beta}(P \| Q),
$$
for $0<\alpha<1<\beta<\infty$, where the logarithm used in the definitions of Kullback-Leibler and Rényi are with respect base $e$. Therefore, their quantum counterparts also  satisfy the inequalities
\begin{align}
\mathscr{H}_\alpha(\rho \| \sigma) &\log e \leq D_\alpha(\rho \| \sigma) \leq D(\rho \| \sigma) \leq D_{\beta}(\rho \| \sigma)\leq \mathscr{H}_{\beta}(\rho \| \sigma),
 \end{align} for $0<\alpha<1<\beta<\infty$.

\item \label{item:10}  By \cite[Theorem 3.1]{csiszar1972} we get the following\footnote{By setting $m=1$ and $w_1=1$ in \cite[Theorem 3.1]{csiszar1972} and using Remark \ref{rmk:f-div-1} we obtain that $Q^* = P$, where $Q^*$ is as in \cite{csiszar1972}.}: Let $f:(0, \infty) \rightarrow [0,\infty)$ be a strictly convex function such that $f(1)=0$.
Then there exists a real-valued function $\psi_f$ with $\lim _{x \downarrow 0} \psi_f(x)=0$ such that
$$
V(P\|Q) \leq \psi_f\left(D_f(P \| Q)\right).
$$
 This implies that if 
$$
\lim_{n \rightarrow \infty} D_f\left(P_n \| Q_n\right)=0 \Rightarrow \lim _{n \rightarrow \infty} V(P_n\| Q_n)=0.
$$ 
The assumptions on $f$ are valid for the function giving rise to the squared Hellinger distance and $\chi^2$-relative entropy.
Therefore, 
if $f:(0, \infty) \rightarrow [0,\infty)$ is a strictly convex function such that $f(1)=0$, 
then there exists a real-valued function $\psi_f$ such that $\lim _{x \downarrow 0} \psi_f(x)=0$ and 
$$
V(\rho\|\sigma) \leq \psi_f\left(D_f(\rho \| \sigma)\right),
$$
which implies \begin{align}
    \label{item:10.1} 
    \lim_{n \rightarrow \infty} D_f\left(\rho_n \| \sigma_n\right)=0 \Rightarrow \lim _{n \rightarrow \infty} V(\rho_n\|\sigma_n)=0.
\end{align}
\end{enumerate}
\section{Conclusion and Discussion}
The major contribution of this article is a framework for obtaining quantum versions of several  results available for classical $f$-divergences. This is illustrated by proving quantum versions of several important classical $f$-divergence inequalities. A comprehensive review of the use of Nussbaum-Szkoła distributions in the literature is provided in Section \ref{sec:review}. We believe that the review will be useful for researchers working on related areas. Further applications of our main theorem are provided in the subsequent articles \cite{androulakis-john-2022b, Androulakis-John-2023}. All our results work both in finite and infinite dimensional setting which is a strength of the methods adopted in this article. Hence these results are particularly useful in continuous variable quantum information theory as well. For instance, we use the results of this article in order to study the Petz-R\'enyi relative entropy of gaussian states in the follow-up article \cite{Androulakis-John-2023}. 
\appendix
\section*{Appendix}

\section{Classical \texorpdfstring{\ensuremath{f}}{}-divergence}\label{appendix:classical-divergences}
In this section, we recall a few facts about the $f$-divergences in the setting of classical probability. We refer to \cite{Liese-Vajda-2006} and the survey article \cite{Erven-Harremos-2014} for the following definitions and results which we state in this section. The results  from \cite{Erven-Harremos-2014} which we use in this article  are repeated here for the ease of the reader. If $\mu$ and $\nu$ are two positive measures on a measure space $(X,\mathcal{F})$, then $\nu$ is said to be absolutely continuous with respect to $
\mu$ and we write $\nu\ll \mu$, if for every $E\in \mathcal{F}$ such that $\mu (E) = 0$, then $\nu(E) = 0$.
\begin{defn}
 \label{defn:classical-f-divergence}\cite[equation 24]{Liese-Vajda-2006}. Let $P, Q$ be probability distributions on a measure space $(X, \mathcal{F})$. Let $\mu$ be any $\sigma$-finite measure such that $P\ll\mu$ and $Q\ll\mu$. Let $p$ and $q$  denote the Radon-Nikodym derivatives with respect to $\mu$, of $P$ and $Q$, respectively.
 Let $f:(0,\infty)\to\BR$ be a convex (or concave) function then the $f$-divergence $D_f(P||Q)$ is defined as 
 \begin{align}
    D_f(P||Q) = \int\limits_{X}qf\left(\frac{p}{q}\right)\dd\mu,
 \end{align}
 under the conventions in \eqref{eq:convex-f-notations-1} and \eqref{eq:convex-f-notations-2}. If $P\ll Q$ \begin{equation}
     \label{eq:f-div}
      D_f(P||Q) = \int\limits_{X}f\left(\frac{p}{q}\right)\dd Q,
 \end{equation} where a proof of this fact can found in \cite[Page 4398]{Liese-Vajda-2006}. For completeness, we will prove it in the case of discrete probability distributions in  Remark~\ref{defn:renyi-divergence-Kullback}.
\end{defn}
\begin{rmk}\label{rmk:f-div-1}
If $f(1) = 0$ then $D_f(P\|P)=0$ for all probability distributions $P$. If $f(1) = 0$ and $P\ll Q$ then  by Jensen's inequality it can be seen that \[D_f(P\|Q)\geq 0.\]
Furthermore, if $f\geq 0$, then $D_f(P\|Q)\geq 0$ without any additional assumptions. Thus if  $f(1) = 0$ and $f\geq 0$, then  \[\min_{Q}D_f(P\|Q) = D_f(P\|P)=0. \]
\end{rmk}
\begin{rmk}\label{defn:renyi-divergence-Kullback}
  We refer to \cite{Liese-Vajda-2006} for more details on the definition and the properties of general classical $f$-divergences. Here we describe briefly about the case when the measures are discrete.  Let $P$ and $Q$ be discrete probability measures defined on a countable set $\mathcal{I}$. Let $\mu$ be the counting measure on $\mathcal{I}$, then clearly $P\ll\mu$, $Q\ll\mu$, \begin{align*}
     \dv{P}{\mu}\,(i)=P(i), \quad \textnormal{and} \quad \dv{Q}{\mu}\,(i)=Q(i), \forall i\in \mathcal{I}
 \end{align*} where $\dv{P}{\mu}$ and $\dv{Q}{\mu}$ are the respective Radon-Nikodym derivatives. In this case, \begin{equation}\label{f-div-discrete}
     D_f(P||Q) =\sum_{i\in \mathcal{I}}Q(i)f\left(\frac{P(i)}{Q(i)}\right),
 \end{equation}
 under the conventions in \eqref{eq:convex-f-notations-1} and  \eqref{eq:convex-f-notations-2}. We will show now that the sum in the equation above is independent of rearrangements. Assume with out loss of generality that $f$ is a convex function. We will show in this case that the negative terms of the series above sum up to a finite number. This fact, once proved, ensures that the series is independent of rearrangements. Let \[\mathcal{I}_- = \left\{i\in \mathcal{I}: \, Q(i)f\left(\frac{P(i)}{Q(i)}\right)<0\right\}.\] Since $f$ is convex, there exists $a, b \in \BR$ such that \[f(x)\geq a+bx, \quad \forall x\in [0,\infty),\]
where $f(0)=\lim_{x\downarrow0}f(x)$.
 For $i \in \mathcal{I}_-$  such that $Q(i)=0$ we have by our conventions in 
 \eqref{eq:convex-f-notations-2} that 
\begin{align*}f\left(\frac{P(i)}{Q(i)}\right) Q(i) &= P(i) \lim _{s \rightarrow \infty} \frac{f(s)}{s} \\ 
&\geqslant
\quad P(i) \lim _{s \rightarrow \infty} \frac{a+bs}{s}\\
&=P(i) b \\
&\geqslant-\abs{a} Q(i)-\abs{b} P(i).
\end{align*}
For $i \in \mathcal{I}_-$ such that  $Q(i) \neq 0$ we have \begin{align*} f\left(\frac{P(i)}{Q(i)}\right) Q(i) &\geqslant\left(a+b \frac{P(i)}{Q(i)}\right) Q(i)\\
& =a Q(i)+b P(i)\\&\geqslant-|a| Q(i)-|b| P(i).\end{align*}
So \[\sum_{i \in \mathcal{I}_{-}} f\left(\frac{P(i)}{Q(i)}\right) Q(i) \geqslant -\sum_{i \in \mathcal{I}_-}(\abs{a} Q(i)+\abs{b} P(i)) \geqslant-|a|-|b|>-\infty.\]
Now because rearrangements are possible, we may compute an expression for $D_f(P\|Q)$ using the conventions given in \eqref{eq:convex-f-notations-1} and \eqref{eq:convex-f-notations-2} as follows: 
\begin{align*}
D_f(P \| Q)&=\sum_i Q(i) f\left(\frac{P(i)}{Q(i)}\right) \\
& =\sum_{\{i: P(i)Q(i)>0\} } Q(i) f\left(\frac{P(i)}{Q(i)}\right)+ \sum_{Q(i)\neq 0, P(i)=0}f(0)Q(i)+ \sum_{Q(i)= 0, P(i)\neq 0}P(i)f'(\infty) \\
& =\sum_{\{i: P(i)Q(i)>0\} } Q(i) f\left(\frac{P(i)}{Q(i)}\right)+f(0) Q(P=0)+f^{\prime}(\infty) P(Q=0).\numberthis \label{eq:df-pq-expression} 
\end{align*}
It may be noted that, if $P\ll Q$, then the last term above is not present and we obtain \eqref{eq:f-div}.
\end{rmk}
\begin{egs}\label{eg:classical-examples} Here we list a few important examples of $f$-divergences. While the definitions are general, for the purpose of displaying an expression for each these quantities in our context, we assume that $P$ and $Q$ are discrete as in Remark \ref{defn:renyi-divergence-Kullback} and use Equation~\eqref{eq:df-pq-expression}. 
\begin{enumerate}
    \item $f(t)=t \log t \Rightarrow D(P\|Q):= D_{f}(P\|Q)$, the \textbf{Kullback-Leibler Relative Entropy},
    \begin{align*}
    D(P\|Q)=   \begin{cases}
        \sum_{i\in \mathcal{I}} P(i)\log \frac{P(i)}{Q(i)}, & \text{if } P\ll Q;\\
        \infty, &\text{otherwise.}
    \end{cases} 
    \end{align*}
    \item For $\alpha\in (0,1)\cup (1,\infty)$, $f_\alpha(t)=t^\alpha$ we obtain   $D_{\alpha}(P||Q) := \frac{1}{\alpha-1} \log D_{f_\alpha}(P\|Q)$, the  \textbf{R\'enyi} $\bm{\alpha}$-\textbf{relative entropy},
    \begin{align*}
       D_{\alpha}(P||Q)  = \begin{cases}
           \frac{1}{\alpha-1} \log \sum_{i\in \mathcal{I}}P(i)^{\alpha}Q(i)^{1-\alpha}, &\text{if } \alpha<1, 
            \text{ or }
            P\ll Q;\\
            \infty, & \text{otherwise}.
        \end{cases}
    \end{align*}
    \item \label{item:hellinger} $f(t) = \frac{1}{2}(\sqrt{t}-1)^2 \Rightarrow \mathscr{H}^2(P||Q):= D_f(P||Q)$, the \textbf{squared Hellinger distance}, \begin{align*}
        \mathscr{H}^2(P||Q) = \frac{1}{2}\sum_{i\in \mathcal{I}}\left(\sqrt{P(i)}-\sqrt{Q(i)}\right)^2.
    \end{align*} 
    \item For $\alpha\in (0,1)\cup (1,\infty)$, $f_{\alpha}(t)=\frac{t^\alpha-1}{\alpha-1} 
$ we obtain   $ \mathscr{H}_{\alpha}(P||Q):=D_{f_\alpha}(P \| Q)$, the \textbf{Hellinger $\alpha$-divergence},
\[\mathscr{H}_{\alpha}(P||Q)=\begin{cases}
     \frac{1}{\alpha-1}\left( \left(\sum_{i}P(i)^{\alpha}Q(i)^{1-\alpha}\right)-1\right), & \text{if } \alpha <1 \text{ or } P\ll Q;\\
     \infty, &
     \text{otherwise.}
\end{cases}\]
\item  $f(t)=|t-1| \Rightarrow V(P\|Q):= D_f(P \| Q)$, the \textbf{Total Variation Distance},
\begin{align*}
  V(P\|Q) = \sum_{i\in \mathcal{I}}\abs{P(i)-Q(i)}.  
\end{align*}
\item  $f(t)=(t-1)^2 \Rightarrow \chi^2(P||Q):= D_f(P \| Q)$, the $\bm{\chi^2}$-\textbf{divergence},
\begin{align*}
    \chi^2(P||Q) = \begin{cases}
        \sum_{\{i\in \mathcal{I}|Q(i)>0\}} \frac{\left(P(i)-Q(i)\right)^2}{Q(i)},& \text{if } P\ll Q;\\
        \infty, &\text{otherwise.}
    \end{cases}
\end{align*}
\end{enumerate} 
\end{egs}

\section{Relative Modular Operator on \texorpdfstring{\ensuremath{\B{\CK }}}{}}\label{appendix:relative-modular}Let $\rho$ and $\sigma$ be two density operators  on a Hilbert space $\CK $. 
In this section, we analyse the relative modular operator $\Delta_{\rho,\sigma}$ and find its spectral decomposition in our setting. One can refer to  \cite{araki1977relative, petz-ohya-1993, luczak2019relative}  for studying it in the general von~Neumann algebra setting. Let $(\Bt{\CK }, \braket{\cdot}{\cdot}_2)$ denote the Hilbert space of Hilbert-Schmidt operators on $\CK $. Let $\Pi_{\sigma}$ denote the orthogonal projection onto the support of $\sigma$. Define   \[D(S)= \{X\sqrt{\sigma}\,:\,X\in \B{\CK }\}+ \{Y(I-\Pi_{\sigma})\,:\, Y\in \Bt{\CK}\}\subseteq \Bt{\CK },\]
which is a (vector space) direct sum of orthogonal linear manifolds of $\Bt{\CK}$. Then $D(S)$ is a dense linear submanifold.
Now define the  antilinear operator $S:D(S)\rightarrow \Bt{\CK}$ 
by \begin{align}\label{eq:antilinear-S}
    S\left(X\sqrt{\sigma} + Y(I-\Pi_{\sigma})\right) =\Pi_{\sigma}X^{\dagger}\sqrt{\rho}.
\end{align}
For the densely defined antilinear operator $S$, the adjoint $S^{\dagger}$ is defined on all $a\in \Bt{\CK}$ such that there exists a vector $S^{\dagger}a\in \Bt{\CK}$ satisfying \[\braket{a}{Sb} = \overline{\braket{S^{\dagger}a}{b}}, \quad \forall b\in D(S).\]
 Now let $X_2\in \B{\CK}$ and $Y_2\in \Bt{\CK}$ be arbitrary, for the antilinear operator $S$ defined by (\ref{eq:antilinear-S}), we have 
 \begin{align*}
     &\braket{\sqrt{\sigma}X_2+(I-\Pi_{\sigma})Y_2}{S\left(X_1\sqrt{\sigma} + Y_1(I-\Pi_{\sigma})\right)}_2\\
     &\phantom{.................................}=  \braket{\sqrt{\sigma}X_2+(I-\Pi_{\sigma})Y_2}{\Pi_{\sigma}X_1^{\dagger}\sqrt{\rho}}_2\\
     &\phantom{.................................}= \tr \left\{\left(X_2^{\dagger}\sqrt{\sigma}+Y_2^{\dagger}(I-\Pi_{\sigma})\right)\Pi_{\sigma}X_1^{\dagger}\sqrt{\rho}\right\}\\
     &\phantom{.................................}= \tr  X_2^{\dagger}\sqrt{\sigma}\Pi_{\sigma}X_1^{\dagger}\sqrt{\rho}\\
     &\phantom{.................................}= \tr X_2^{\dagger}\Pi_{\sigma}\sqrt{\sigma}X_1^{\dagger}\sqrt{\rho}\\
     &\phantom{.................................}= \tr \sqrt{\sigma}X_1^{\dagger}\sqrt{\rho} X_2^{\dagger}\Pi_{\sigma}\\
     &\phantom{.................................}= \tr\left\{ \left(\sqrt{\sigma}X_1^{\dagger}+(I-\Pi_\sigma)Y_1^{\dagger}\right)\sqrt{\rho} X_2^{\dagger}\Pi_{\sigma}\right\}\\
     &\phantom{.................................}=\braket{X_1\sqrt{\sigma}+Y_1(I-\Pi_\sigma)}{\sqrt{\rho} X_2^{\dagger}\Pi_{\sigma}}_2\\
      &\phantom{.................................}=\overline{\braket{\sqrt{\rho} X_2^{\dagger}\Pi_{\sigma}}{X_1\sqrt{\sigma}+Y_1(I-\Pi_\sigma)}_2}.
 \end{align*}
 Hence  $S^{\dagger}$ is defined of the dense set \[\left\{\sqrt{\sigma}X_2+(I-\Pi_{\sigma})Y_2\,:\, X_2\in \B{\CK}, Y_2\in \Bt{\CK} \right\},\]
and 
\begin{align}\label{eq:antilinear-S-dagger}
    S^{\dagger}\left(\sqrt{\sigma}X_2+(I-\Pi_{\sigma})Y_2\right) = \sqrt{\rho} X_2^{\dagger}\Pi_{\sigma}, \quad \forall X_2\in \B{\CK}, Y_2\in \Bt{\CK}.
\end{align}
Since $S$ and $S^{\dagger}$ are densely defined, $S$ is a closable operator.
The relative modular operator $\Delta_{\rho,\sigma}$ is defined as \begin{align}
    \label{eq:modular-operator}
    \Delta_{\rho,\sigma} =S^{\dagger}\overline{S}
\end{align}
where $\overline{S}$ denotes the closure of $S$. Furthermore, we have \begin{equation}
    \label{eq:dom-delta}
  \{  X\sigma+Y(I-\Pi_{\sigma})\,:\, X\in \B{\CK}, Y\in \Bt{\CK}\}\subseteq D(\Delta_{\rho,\sigma})
\end{equation} and
\begin{align*}\label{eq:modular-operator-action}
    \Delta_{\rho,\sigma} \left(X\sigma+Y(I-\Pi_{\sigma})\right) &= S^{\dagger}\overline{S}\left(X\sqrt{\sigma}\sqrt{\sigma}+Y(I-\Pi_{\sigma}\right)\\
    &= S^{\dagger}\left(\Pi_{\sigma}\sqrt{\sigma}X^{\dagger}\sqrt{\rho}\right)\\
     &= S^{\dagger}\left(\sqrt{\sigma}\Pi_{\sigma}X^{\dagger}\sqrt{\rho}\right)\\
    &= \sqrt{\rho}\left(\Pi_{\sigma}X^{\dagger}\sqrt{\rho}\right)^\dagger\Pi_{\sigma}\\
    &=\rho X\Pi_{\sigma}, \quad \forall X\in \B{\CK },  Y\in \Bt{\CK} \numberthis.
\end{align*}
\begin{prop}\label{prop:spectral-relative-modular}
Let $\rho$ and $\sigma$ be as in \eqref{eq:spectral-rho-sigma}. Then the spectral decomposition of the relative modular operator $\Delta_{\rho,\sigma}$ is given by
\begin{align}\label{eq:delta-spectral}
    \Delta_{\rho,\sigma} = \sum_{\{i,j\,:\, r_i\neq 0, s_j \neq 0\}} r_is_j^{-1}\ketbra{X_{ij}},
\end{align}
where \begin{align}
    \label{eq:X-ell-j-defn}
    X_{ij} = \ketbra{u_i}{v_j} \in \Bt{\CK},\quad \forall i,j\in \mathcal{I}.
\end{align}In particular, \[\ker\Delta_{\rho,\sigma} = \overline{\spn}\{X_{ij}\,:\,r_i=0 \textnormal{ or } s_j =0\}.\]
\end{prop}
\begin{proof}
Since $\{u_i\}$ and $\{v_j\}$ are orthonormal bases for $\CK $, it is easy to see that the double sequence \[\{X_{ij}\}_{i, j\in \mathcal I}\text{ is an orthonormal basis of } \Bt{\CK }. \]
To complete the proof, we will show the following: \begin{enumerate}
    \item \label{item:counter-eg-1} $X_{ij}\in D(\Delta_{\rho,\sigma})$ for all $i,j$;
    \item\label{item:counter-eg-3} $\Delta_{\rho,\sigma} (X_{i,j}) = r_is_j^{-1} X_{ij}$ for all $i,j$ such that  $s_j\neq 0$;
    \item \label{item:counter-eg-2} if $s_j =0$ then $X_{ij}\in \ker \Delta_{\rho,\sigma}$.
\end{enumerate}
To prove \ref{item:counter-eg-1}, note that for $j$ such that  $s_j\neq 0$, by \eqref{eq:dom-delta}, \begin{align}\label{eq:delta-spectral-proof-1}
    X_{ij} = \ketbra{u_i}{v_j} =s_j^{-1}\ketbra{u_i}{v_j}\left(\sum_{k}s_k\ketbra{v_k}\right)= s_j^{-1}X_{ij}\sigma\in D(\Delta_{\rho,\sigma}), \forall i.
\end{align}
Also, if $s_j=0$ then $v_j\in \ran(I-\Pi_{\sigma})$ and once again from \eqref{eq:dom-delta}, \begin{align}
    \label{eq:delta-spectral-proof-2}X_{ij} = \ketbra{u_i}{v_j}= \ketbra{u_i}{v_j}(I-\Pi_{\sigma})\in D(\Delta_{\rho,\sigma}), \forall i.
\end{align}
Now by \eqref{eq:modular-operator-action} and \eqref{eq:delta-spectral-proof-1}, for $j$ such that  $s_j\neq 0$ \[\Delta_{\rho,\sigma}X_{ij} = \Delta_{\rho,\sigma}(s_j^{-1}X_{ij}\sigma)=\rho (s_j^{-1}X_{ij})\Pi_{\sigma}=s_j^{-1}\left(\sum_{k}r_k\ketbra{u_k}\right)\ketbra{u_i}{v_j}\Pi_{\sigma}=r_is_j^{-1}X_{ij}, \]
which proves \ref{item:counter-eg-3}. Item \ref{item:counter-eg-2}  follows from \eqref{eq:modular-operator-action} and \eqref{eq:delta-spectral-proof-2}.
\end{proof}
 \begin{rmk}
 The spectral projections of  general modular operator with faithful states can be seen in \cite[equation 2.9]{Datta-Pautrat-Rouze-2016} for the finite dimensional case and \cite[Example 2.6]{Hiai-2018} for general dimensions. Our Proposition \ref{prop:spectral-relative-modular} provides the spectral decomposition of the relative modular operator even when the states are not faithful.
 \end{rmk}

\section*{Acknowledgements}
We thank Mark Wilde, Mil{\'a}n Mosonyi and Hemant Kumar Mishra for their comments on an earlier version \cite{Androulakis-John-2022a} of this article, and for the references that they brought to our attention. We also thank the anonymous referee for several constructive suggestions which
helped us to improve the article immensely.

The second author thanks The Fulbright Scholar Program and United States-India Educational Foundation for providing funding and other support to conduct this research through a Fulbright-Nehru Postdoctoral Fellowship (Award No. 2594/FNPDR/2020), he also acknowledges the United States Army Research Office MURI award on
Quantum Network Science, awarded under grant number W911NF2110325 for partially funding this research.

\bibliographystyle{IEEEtran}
\bibliography{bibliography}
\end{document}